\documentclass[a4paper,UKenglish,cleveref, autoref, thm-restate]{lipics-v2021}
\title{Geodetic Set on Graphs of Constant 
Pathwidth and Feedback Vertex Set Number} %TODO Please add

\titlerunning{Geodetic Set on Graphs of Constant 
Pathwidth and Feedback Vertex Set Number} %TODO optional, please use if title is longer than one line

\author{Prafullkumar Tale}{Indian Institute of Science Education and Research Pune, Pune, India \and \url{https://pptale.github.io/}}{prafullkumar@iiserb.ac.in}{https://orcid.org/0000-0001-9753-0523}{}

\authorrunning{P. Tale} %TODO mandatory. First: Use abbreviated first/middle names. Second (only in severe cases): Use first author plus 'et al.'

\Copyright{Prafullkumar Tale} %TODO mandatory, please use full first names. LIPIcs license is "CC-BY";  http://creativecommons.org/licenses/by/3.0/


 %TODO mandatory; please add comma-separated list of keywords

\nolinenumbers %uncomment to disable line numbering

\hideLIPIcs

\EventEditors{John Q. Open and Joan R. Access}
\EventNoEds{2}
\EventLongTitle{42nd Conference on Very Important Topics (CVIT 2016)}
\EventShortTitle{CVIT 2016}
\EventAcronym{CVIT}
\EventYear{2016}
\EventDate{December 24--27, 2016}
\EventLocation{Little Whinging, United Kingdom}
\EventLogo{}
\SeriesVolume{42}
\ArticleNo{23}
\usepackage{complexity}
\usepackage{mathtools}
\usepackage{bm}
\usepackage{xspace}
\usepackage{enumitem}
\usepackage{mathdots}
\usepackage{float}
\usepackage{needspace}
\usepackage[framemethod=tikz]{mdframed}

\newcommand{\para}{\textsf{para}}

\newcommand{\calU}{\mathcal{U}}
\newcommand{\calX}{\mathcal{X}}

\newcommand{\mdfull}{\textsc{Metric Dimension}\xspace}

\newcommand{\gsfull}{\textsc{Geodetic Set}\xspace}

\newcommand{\calS}{\mathcal{S}}
\newcommand{\calO}{\mathcal{O}}

\DeclareMathOperator{\dist}{dist}

\newcommand{\fvs}{\mathtt{fvs}}
\newcommand{\td}{\mathtt{td}}
\newcommand{\tw}{\mathtt{tw}}
\newcommand{\pw}{\mathtt{pw}}
\newcommand{\diam}{\mathtt{diam}}
\newcommand{\yes}{\textsc{Yes}}

\newcommand{\ETH}{\textsf{ETH}}

\newcommand{\pndt}{\textsf{pndt}}

\newtheoremstyle{noparentheses}
    {}{}{\itshape}{}%
    {\bfseries}{.}{ }%
    {\thmname{#1}\thmnumber{ #2}\thmnote{ {\mdseries #3}}}
\theoremstyle{noparentheses} 
\newtheorem*{theoremnp*}{Theorem}

\usepackage{color}

\usepackage[colorinlistoftodos,textsize=tiny]{todonotes}
\usepackage{tikz}
\usetikzlibrary{positioning,backgrounds,patterns,calc}

\tikzset{
  circ/.style = {circle,draw,fill,inner sep=1.3pt}
}

\newtheorem{reduction rule}{Reduction Rule}
\begin{document}

\maketitle

\begin{abstract}
% !TEX root = ./main.tex
In the \textsc{Geodetic Set} problem, the input consists of a graph $G$
and a positive integer $k$.
The goal is to determine whether there exists a subset $S$ of vertices
of size $k$ such that every vertex in the graph is included in a
shortest path between two vertices in $S$.
Kellerhals and Koana [IPEC 2020; J. Graph Algorithms Appl 2022] proved that
the problem is $\W[1]$-hard when parameterized
by the pathwidth and the feedback vertex set number of the input graph.
They posed the question of whether the problem admits an
$\XP$ algorithm when parameterized by the combination of these two parameters.
We answer this in negative by proving that the problem remains
\NP-hard on graphs of constant pathwidth
and feedback vertex set number.
\end{abstract}

%\tableofcontents

% !TEX root = ./main.tex 
\section{Introduction}
\label{sec:intro}

Consider a simple undirected graph $G$ with vertex set $V(G)$ and edge set $E(G)$.
For two vertices $u, v \in V(G)$, let $I(u, v)$ be the set of all vertices in $G$ that are part of some shortest path between $u$ and $v$. For a subset $S$ of vertices, we generalize this definition to $I(S)$ by taking the union of $I(u, v)$ for all pairs of vertices $u, v$ in $S$. A subset of vertices $T$ is said to be \emph{covered} by $S$ if $T \subseteq I(S)$. A set of vertices $S$ is a \emph{geodetic set} if $V(G)$ is covered by $S$. In the \textsc{Geodetic Set} problem, the input is a graph $G$ and a positive integer $k$, and the objective is to determine whether there is a geodetic set of size $k$.
This problem was introduced in $1993$ by Harary, Loukakis, and Tsouros~\cite{harary1993}.
We refer readers to \cite{floCALDAM20,ekim2012} for the application
of the problem.
See~\cite{farber1986,bookGC} for a broader
discussion on geodesic convexity in graphs.

The \textsc{Geodetic Set} problem falls under the broad category of \emph{metric graph problems}. These problems are defined using a metric on the graph, with the shortest distance between two vertices being the most commonly used metric. Metric graph problems include many important
$\NP$-complete graph problems such as \textsc{{Distance $d$}-Dominating Set} (also called \textsc{$(k,d)$-Center}), \textsc{{Distance} $d$-Independent Set} (also called \textsc{$d$-Scattered Set}), and \textsc{Metric Dimension}. These problems have been central to the development of classical as well as parameterized algorithms and complexity theory~\cite{KLP19,KLP22,BelmonteFGR17,LM21,DBLP:conf/wg/BergougnouxDI24,DBLP:conf/fct/ChakrabortyFH23,DBLP:conf/icalp/FoucaudGK0IST24,JKST19,DBLP:conf/isaac/ChakrabortyFMT24}, as they behave quite differently from their more `local' (neighborhood-based) counterparts such as  \textsc{Dominating Set} or \textsc{Independent Set}.
%We refer readers to Section~\ref{sec:preliminaries} for definitions of the terms used here.

Due to the non-local nature of metric graph problems, most algorithmic techniques fail to yield positive results.
Consider the case of the \textsc{Metric Dimension} problem
in which an input is a graph $G$ and an integer $k$,
and the objective is to determine whether there is
a set $S$ of at most $k$ vertices such that for any pair of vertices
$u$ and $v$, there is a vertex in $S$ which has distinct distances
to $u$ and $v$.
In a seminal paper, Hartung and Nichterlein~\cite{HartungN13} proved that the problem is \W[2]-hard when parameterized by the solution size $k$. This motivated the structural parameterization of the problem~\cite[Section 2.1]{DBLP:conf/icalp/FoucaudGK0IST24}.
The complexity of \textsc{Metric Dimension} parameterized by treewidth
remained an open problem for a long time, until Bonnet and Purohit~\cite{BP21}
made the first breakthrough by proving that the problem is \W[1]-hard when
parameterized by pathwidth.
This result was strengthened in two directions: Li and Pilipczuk~\cite{LM21}
proved that it is \para-\NP-hard when parameterized by pathwidth,
whereas Galby et al.~\cite{GKIST23} showed that the problem is
\W[1]-hard when parameterized by pathwidth
plus feedback vertex set number of the graph.

Continuing the `hardness saga', we mention some results that hold
for both \textsc{Metric Dimension} and \textsc{Geodetic Set}.
Foucaud et al.~\cite{DBLP:conf/icalp/FoucaudGK0IST24} demonstrated that
these two problems on graphs of bounded diameter admit
double exponential conditional lower bounds when parameterized by treewidth, along with a matching algorithm.
This was a \emph{first-of-its-kind} result for a natural \NP-complete problem.
The same set of authors also proved that both problems admit `exotic'
conditional lower bounds when parameterized by the vertex cover
number~\cite{DBLP:conf/stacs/FoucaudGK0IST25}.
Bergougnoux et al.~\cite{DBLP:conf/wg/BergougnouxDI24} showed that
enumerating minimal solution sets for both problems in
split graphs is equivalent to enumerating minimal transversals of
hypergraphs, arguably the most important open problem
in algorithmic enumeration.

As illustrated in the examples above, \textsc{Metric Dimension} and
\textsc{Geodetic Set} behave similarly concerning hardness results.
Kellerhals and Koana~\cite{KK22}
proved that \textsc{Geodetic Set} is \W[1]-hard when parameterized by
the pathwidth and feedback vertex set number (and the solution size)
combined, a result similar to that of
Galby et al.~\cite{GKIST23} regarding \textsc{Metric Dimension}.
However, a result similar to that of Li and Pilipczuk's work~\cite{LM21}
on \textsc{Metric Dimension} for \textsc{Geodetic Set} is not known.
More formally, we do not know the complexity of the
\textsc{Geodetic Set} problem on graphs of constant treewidth.
Kellerhals and Koana~\cite{KK22} posed an even stronger question: Does \textsc{Geodetic Set} admit an $\XP$ algorithm when parameterized by the pathwidth
and feedback vertex set number?
We answer this question in the negative.

\begin{theorem} \label{thm:np-hardness}
\textsc{Geodetic Set} is \NP-complete even on graphs of
pathwidth $17$ and feedback vertex set number $13$.
\end{theorem}

Note that this result is surprising in the context of the closely
related problem \textsc{Geodetic Hull}.
In \textsc{Geodetic Hull}, the input is a graph $G$ and an integer $k$, and the objective is to determine whether there is a vertex set $S$ of size $k$ such that $I|V^j(G)|[S] = V(G)$, where $I^0[S] = S$ and $I^j[S] = I[I^{j-1}[S]]$ for $j > 0$.
It is easy to see that both these problems are trivial on trees as it is
necessary and sufficient to take all the leaves into the solution.
Kante et al.~\cite{DBLP:journals/tcs/KanteMS19} proved, among other things, that the problem admits
an \XP-algorithm parameterized by treewidth.
Hence, while \textsc{Geodetic Hull} is polynomial-time solvable on
graphs of constant treewidth, \textsc{Geodetic Set} remains \NP-hard.

\subparagraph*{Related Work}
As is often the case with metric-based problems, \gsfull is
computationally hard, even for very structured graphs. 
See~\cite{atici2002,bueno2018,floISAAC20,floCALDAM20,DBLP:journals/tcs/ChakrabortyGR23,CHZ02,dourado2008,dourado2010,JCMCC96,ekim2012,harary1993}
for various earlier hardness results.
\gsfull can be solved in polynomial time
on split graphs~\cite{dourado2010,JCMCC96}, 
well-partitioned chordal graphs~\cite{wellpart}, outerplanar
graphs~\cite{mezzini2018}, ptolemaic graphs~\cite{farber1986},
cographs~\cite{dourado2010}, distance-hereditary
graphs~\cite{dh}, block-cactus graphs~\cite{ekim2012}, solid grid graphs~\cite{floISAAC20,DBLP:journals/tcs/ChakrabortyGR23}, and proper
interval graphs~\cite{ekim2012}.
A two-player game variant of \gsfull was introduced in~\cite{BH85}.

As mentioned before, Kellerhals and Koana in~\cite{KK22}
studied the parameterized complexity of \gsfull.
They observed that the reduction from~\cite{dourado2010} implies that
the problem is \W[2]-hard when parameterized by the solution size, even for chordal
bipartite graphs.
They proved the problem to be \W[1]-hard for the parameters solution size,
feedback vertex set number, and pathwidth, combined~\cite{KK22}.
On the positive side, they showed that \gsfull is \FPT\ for the parameters treedepth,
modular-width (more generally, clique-width plus diameter), and
feedback edge set number~\cite{KK22}.
Chakraborty et al.~\cite{floISAAC20} proved that the problem is \FPT\ on
chordal graphs when parameterized by the treewidth.
On the approximation algorithms side,
it is known that  the minimization variant of the problem
is \NP-hard to approximate within a factor of $o(\log n)$, even
for diameter~2 graphs~\cite{floCALDAM20} and
subcubic bipartite graphs of arbitrarily large girth~\cite{DIT21}.
%They also proved that it can be approximated in polynomial time within a factor of $n^{1/3}\log n$~\cite{floCALDAM20}.

% !TEX root = ./main.tex
\section{Preliminaries}
\label{sec:preliminaries}
For an integer $n$, we let $[n] = \{1,\ldots,n\}$.
We use standard graph-theoretic notation and refer the reader
to~\cite{D12} for any undefined notation.
For an undirected graph $G$, the sets $V(G)$ and $E(G)$ denote its
set of vertices and edges, respectively.
Two vertices $u,v\in V(G)$ are {\it adjacent} or {\it neighbors} if
$(u, v)\in E(G)$.
We say a vertex $v$ is a \emph{pendant vertex} if its degree is one.
A vertex $v$ is a \emph{branching vertex} if its degree is at least three.
A path is a collection of vertices $\{v_1, v_2, \dots, v_{\ell}\}$
such that $(v_i, v_{i+1})$ for every $i$ in $[\ell - 1]$.
In this case, $v_1$ and $v_{\ell}$ are two endpoints, whereas all the other vertices are called \emph{internal vertices}.
We say a path is simple if all its internal vertices have degree
precisely two.
The length of a path is the number of edges in it.
The {\it distance} between two vertices $u, v$ in $G$, denoted by
$\dist(u,v)$, is the length of a shortest path connecting
$u$ to $v$.
%For a subset $S$ of $V(G)$, we define $N[S] = \bigcup_{v \in S} N[v]$ and $N(S) = N[S] \setminus S$.
For a subset $S$ of $V(G)$, we denote the graph obtained by deleting $S$ from $G$ by $G - S$.
%We denote the subgraph of $G$ induced on the set $S$ by $G[S]$.
Recall that a subset $S \subseteq V(G)$ is a \emph{geodetic set} if for
every $u \in V(G)$, the following holds: there exist $s_1,s_2 \in S$
such that $u$ lies on a shortest path from $s_1$ to $s_2$.
Consider a pendant vertex $v$ in $G$.
Note that $v$ does not belong to any shortest path between
any pair $x,y$ of vertices which are distinct from $v$.
Hence, $v$ belongs to every geodetic set of $G$.
This was also observed in \cite{DBLP:journals/networks/ChartrandHZ02}.

\subparagraph*{Parameterized Complexity.}
{In} parameterized analysis, we associate each instance
$I$ with a parameter $\ell$, and are interested in an algorithm
with running time $f(\ell) \cdot |I|^{\calO(1)}$
for some computable function $f$.
Parameterized problems that admit such an algorithm are called
\emph{fixed parameter tractable} (\FPT) parameterized by $\ell$.
On the other hand, $\W[1]$-hardness categorizes problems that
are unlikely to have \FPT\ algorithms.
A parameterized problem is in \XP~if it admits an algorithm running in time
$|I|^{f(\ell)}$ for some computable function $f$.
We say a parameterized problem in \para-\NP-hard if it is
\NP-hard even when the parameter is a constant.
For more on parameterized complexity and related terminologies, we refer the reader to the recent book by Cygan et al.~\cite{DBLP:books/sp/CyganFKLMPPS15}.

We define some of the parameters mentioned in this article.
For a graph $G$, a set $X \subseteq V(G)$ is a
\emph{feedback vertex set} of $G$ if $G -  X$ is an acyclic graph.
The \emph{feedback vertex set number} of graph $G$, denoted by
$\fvs(G)$, is the size of minimum feedback vertex set of $G$.
A tree decomposition of a graph $G$ is a pair $(\calX , T)$ where
$T$ is a tree and $\calX = \{X_i \mid i \in V (T)\}$ is a collection of subsets of
$V(G)$ such that:
$(i)$ $\bigcup_{i \in V(T)} X_i = V(G)$,
$(ii)$ for each edge $(x, y)$ in $E(G)$, $\{x, y\} \subseteq X_i$
for some $i \in V(T)$, and
$(iii)$ for each $x$ in $V (G)$ the set $\{i \mid x \in  X_i\}$ induces a
connected subtree of $T$.
The width of the tree decomposition is $\textsf{max}_{i\in V (T)}\{|X_i| - 1\}$.
The treewidth of a graph $G$ is the minimum width over all tree
decompositions of $G$.
We denote by $\tw(G)$ the treewidth of graph $G$.
If, in the definition of treewidth, we restrict the tree $T$ to be a path, then we get the notion of pathwidth and denote it by $\pw(G)$.
It is easy to verify that $\tw(G) \le \fvs(G) + 1$ and $\tw(G) \le \pw(G)$
whereas $\pw(G)$ and $\fvs(G)$ are incomparable.

We define an equivalent definition of pathwidth via mixed search games.
In a mixed search game, a graph $G$ is considered a system of tunnels.
Initially, all edges are contaminated by a gas.
An edge is cleared by placing searchers at both its endpoints simultaneously
or by sliding a searcher along the edge.
A cleared edge is re-contaminated if there is a path from
an uncleared edge to the cleared edge without any searchers on its vertices or edges.
A search is a sequence of operations that can be of the following types:
$(i)$ placement of a new searcher on a vertex;
$(ii)$ removal of a searcher from a vertex;
$(iii)$ sliding a searcher on a vertex along an incident edge and
placing the searcher on the other end.
A search strategy is winning if, after its termination, all edges are
cleared.
The mixed search number of a graph $G$, denoted by $\texttt{ms}(G)$, is
the minimum number of searchers required for a winning strategy of mixed searching on $G$.
Takahashi et al.~\cite{DBLP:journals/tcs/TakahashiUK95}
proved that $\pw(G) \le \texttt{ms}(G) \le \pw(G) + 1$,
which we use to bound the pathwidth of the graphs obtained in reduction.
% !TEX root = ./main.tex
\section{NP-Hardness}
\label{sec:np-hardness}

In this section, we prove that the problem is \para-\NP-hard when parameterized
by pathwidth and feedback vertex set number of the input graph.
We present a polynomial-time reduction from the
\textsc{3-Dimensional Matching} problem, which is known to be
\NP-hard~\cite[SP~1]{DBLP:books/fm/GareyJ79}.
For notational convenience, we work with the following definition of the problem.
Input consists of a universe
$\calU = \{\alpha, \beta, \gamma\} \times [n]$,
a family $\calS$ of $m$ subsets of $\calU$ such that
for every set $S \in \calS$, $S = \{(\alpha, a), (\beta, b), (\gamma, c)\}$
for some $a, b, c \in [n]$.
The goal is to find a subset $\calS^{\prime} \subseteq \calS$ of size $n$
that  partitions $\calU$.

\paragraph*{Reduction}
The reduction takes as input an instance $(\calU, \calS)$ of
\textsc{3-Dimensional Matching} and returns an instance
$(G, k)$ of \textsc{Geodetic Set} in polynomial time
where pathwidth plus feedback vertex set number of $G$ is a constant.
Recall that
%$|\calS| = m$, $|\calU| = 3n$ and hence a solution
%is a subset of $\calS$ of cardinality $n$.
%Also,
a {branching vertex} is a vertex of
degree at least three.
We start by specifying the common
branching vertices in $G$ and the lengths of simple paths
connecting them.

\begin{figure}[t]
\centering
\includegraphics[scale=0.5]{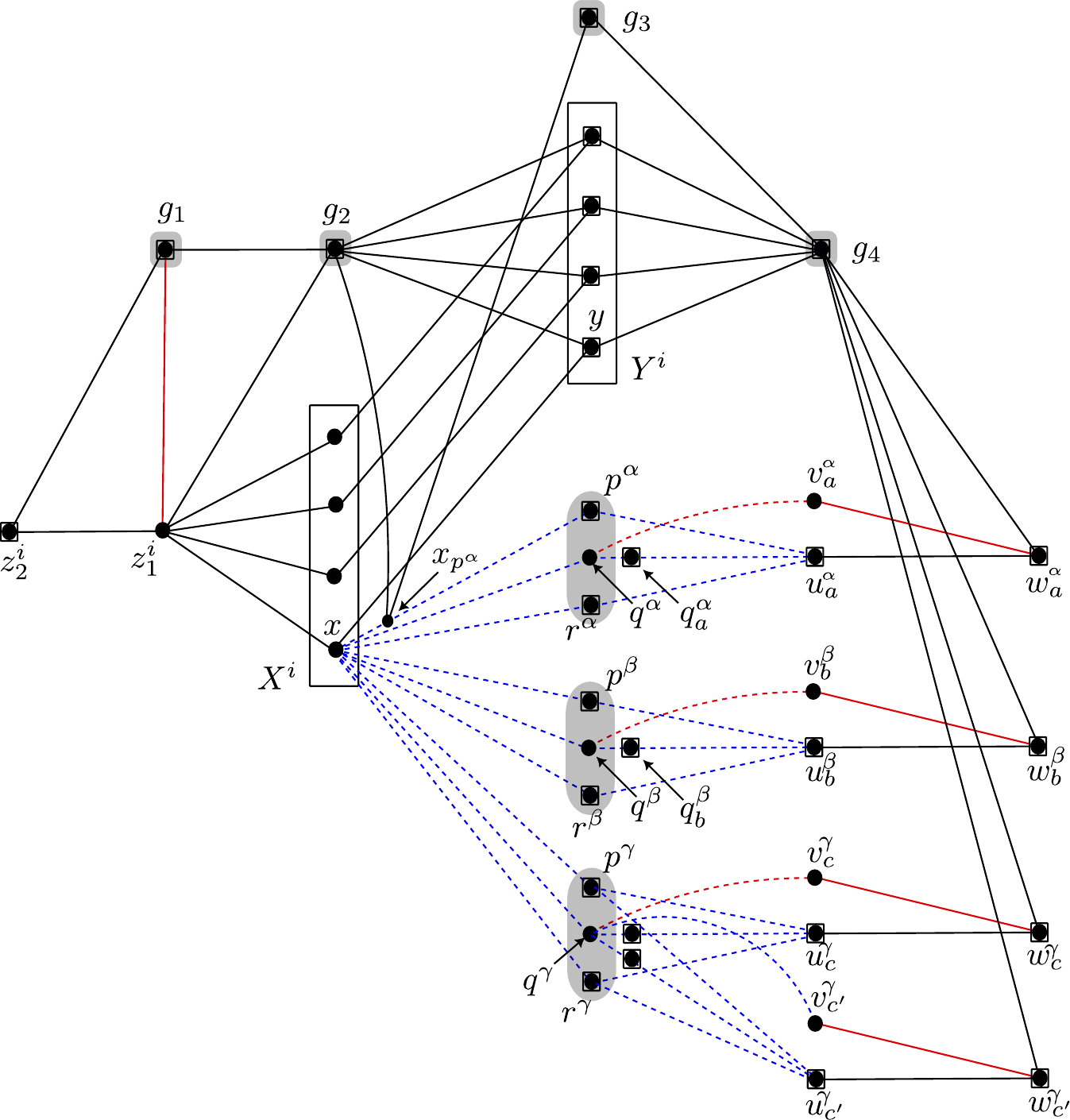}
\caption{The figure shows the $i^{th}$ copy of set-encoding gadget
$\{z^i_2, z^i_1\} \cup X^i \cup Y^i$,
element-encoding gadgets encoding elements of the form
$(\alpha, a)$, $(\beta, b)$,
$(\gamma, c)$ and $(\gamma, c')$, for some $a, b, c, c' \in [n]$,
and the common branching vertices.
Each solid (black or red) edge represents a path of fixed length
which is either $M^2$ or $M^2 - 1$.
Each dotted (blue or red) edge represents a path whose length depends
on the set or element corresponding to one of its endpoints.
Apart from the internal vertices in the red edges,
all the other vertices are covered by pendant vertices in the graph.
For clarity, we show connecting paths
starting at only one vertex in $X^i$.
Moreover, it only shows one of the nine vertices, viz $x_{p^{\alpha}}$,
that are adjacent with $x$ and are connected to $g_2$ and $g_3$.}
\label{fig:np-hardness}
\end{figure}

\medskip
\noindent \textbf{Common Branching Vertices.}
The reduction starts constructing $G$ by adding the following
vertices: $\{g_1, g_2, g_3, g_4\}$,
$\{p^{\alpha}, q^{\alpha}, r^{\alpha} \}$,
$\{p^{\beta}, q^{\beta}, r^{\beta} \}$, and
$\{p^{\gamma}, q^{\gamma}, r^{\gamma}\}$.
All set-encoding gadgets and element-encoding gadgets,
which we define below, are connected via these common branching vertices.
These common branching vertices are denoted using grey shaded
region in all the figures.
Apart from the vertices in $\{q^{\alpha}, q^{\beta}, q^{\gamma}\}$,
the reduction adds pendant vertices adjacent to each common
branching vertex.
In all the figures, vertices that are adjacent to a pendant vertex
are highlighted by enclosed squares around them.
Define an integer $M$ such that $M \in \calO(n^2)$ and
for any $a \in [n]$,
we have $0.99 \cdot M^2 \le M^2 - 2aM, M^2 + 2aM \le 1.01 \cdot M^2$.
The reduction connects $g_1$ to $g_2$ and $g_3$ to $g_4$
using a path of length $M^2$ for each connection.
See Figure~\ref{fig:np-hardness}.

\medskip

\noindent \textbf{Set-encoding Gadget.}
The reduction adds $n$ identical set-encoding gadgets denoted by $\calS^1, \calS^2, \dots, \calS^n$.
For any $i \in [n]$, set-encoding gadget $\calS^i$ contains
set of vertices $X^i$, $Y^i$, and two auxiliary vertices $z^i_1$ and $z^i_2$.
Sets $X^i$ and $Y^i$ contains $m$ vertices each.
The reduction adds matching edges across $X^i$ and $Y^i$.
Each of these matching edges and their endpoints corresponds to
a unique set in $\calS$.
Then, the reduction replaces each of these matching edges
with a simple path of length $M^2$.
For each vertex in $Y^i$, the reduction adds a pendant vertex
adjacent to it.
It connects $z^i_1$ with every vertex in $\{z^i_2\} \cup X^i$
using a simple path of length $M^2$ for each connection.
%See Figure~\ref{fig:encoding-sets}.

%\begin{figure}[t]
%    \centering
%        \includegraphics[scale=0.6]{./images/encoding-sets.pdf}
%    \caption{The figure shows two copies of set-encoding gadgets and common branching vertices.
%    Any set-encoding gadget or element-encoding gadget is
% is attached only to common branching vertices in the grey (shaded) regions.
%    Vertices that are adjacent to pendant vertices are highlighted by
%    enclosed square around them. Each edge represents a path
%    of length $M^2$.}
%    \label{fig:encoding-sets}
%\end{figure}

We now specify how the vertices mentioned above are
connected to common branching vertices.
The reduction connects
$(i)$ $g_1$ to $z^i_1$ and $z^i_2$,
$(ii)$ $g_2$ to $z^i_1$, and
$(iii)$ $g_4$ to every vertex in $Y^i$,
using paths of length $M^2$ for each connection.
It connects $g_2$ to every vertex in $Y^i$ using a path of length
$M^2 - 1$.
We remark that these are the only paths that are of length
$M^2 - 1$.
To specify the connection to the remaining common branching vertices,
consider set $S = \{(a', \alpha), (b', \beta), (c', \gamma)\}$ in $\calS$
and the vertex $x$ corresponding to it in $X^i$.
The reduction connects $x$  to all the remaining
common branching vertices  such that
\begin{itemize}[nolistsep]
\item $\dist(x, p^{\alpha}) = M^2 + 2a'M$,
$\dist(x, q^{\alpha}) = M^2 - a'M$, and
$\dist(x, r^{\alpha}) = M^2 - 2a'M$;
\item $\dist(x, p^{\beta}) = M^2 + 2b'M$,
$\dist(x, q^{\beta}) = M^2 - b'M$, and
$\dist(x, r^{\beta}) = M^2 - 2b'M$; and
\item $\dist(x, p^{\gamma}) = M^2 + 2c'M$,
$\dist(x, q^{\gamma}) = M^2 - c'M$, and
$\dist(x, r^{\gamma}) = M^2 - 2c'M$.
\end{itemize}
See Figure~\ref{fig:np-hardness-distance} for an illustration.

Now, consider the vertex on the path connecting $x$ to $p^{\alpha}$
which is adjacent to $x$.
We denote this vertex by $x_{p^{\alpha}}$.
The reduction connects $x_{p^{\alpha}}$ with $g_2$ and $g_3$ using
the paths of length $M^2$ each.
It repeats this process to add vertices
$x_{q^{\alpha}}$, $x_{r^{\alpha}}$,
$x_{p^{\beta}}$, $x_{q^{\beta}}$, $x_{r^{\beta}}$,
$x_{p^{\gamma}}$, $x_{q^{\gamma}}$, and $x_{r^{\gamma}}$,
and their connection to $g_2$ and $g_3$.

\begin{figure}[t]
\centering
\includegraphics[scale=0.75]{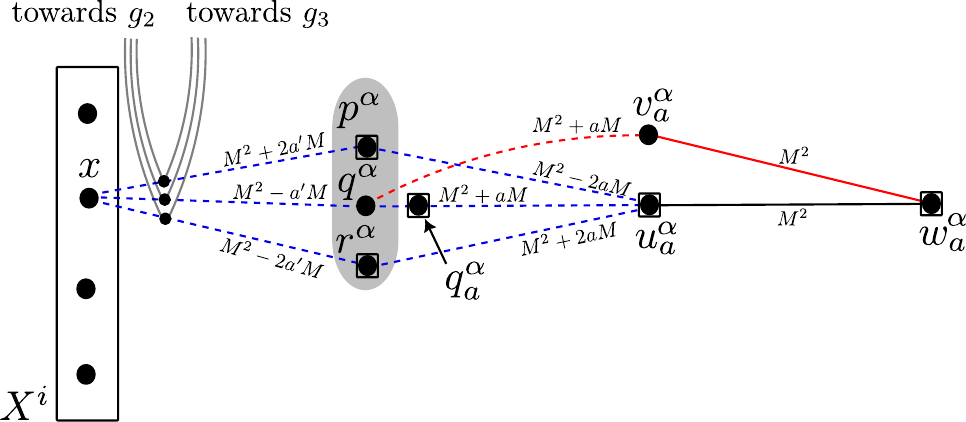}
\caption{On the left of the grey (shaded) area:
Connection between some common branching vertices and a vertex $x$
in $X^i$
which corresponds to set $S = \{(\alpha, a'), (\beta, b'), (\gamma, c')\}$,
for some $a', b', c' \in [n]$.
The figure also shows vertices $x_{p^{\alpha}}$,
$x_{q^{\alpha}}$, and $x_{r^{\alpha}}$.
On the right of the grey (shaded) area: Connection between some common
vertices and element-encoding gadget encoding $(\alpha, a)$
for some $a \in [n]$.}
\label{fig:np-hardness-distance}
\end{figure}

\medskip

\noindent \textbf{Element Encoding Gadget.}
Consider an element $(\alpha, a)$ in $\calU$.
The reduction adds three vertices $u^{\alpha}_a, v^{\alpha}_a$
and  $w^{\alpha}_a$.
It adds a pendant vertex adjacent to $u^{\alpha}_a$ and another
pendant vertex adjacent to $w^{\alpha}_a$.
It connects $w^{\alpha}_a$ with $u^{\alpha}_a$ and $v^{\alpha}_a$
using a simple path of length $M^2$ for
each connection.

The reduction connects 
$u^{\alpha}_a$ with $\{p^{\alpha}, q^{\alpha}, r^{\alpha}\}$
using simple paths such that
\begin{itemize}[nolistsep]
\item $\dist(u^{\alpha}_a, p^{\alpha}) = M^2 - 2aM$,
$\dist(u^{\alpha}_a, q^{\alpha})= M^2 +aM$, and
$\dist(u^{\alpha}_a, r^{\alpha})= M^2 + 2aM$.
\end{itemize}
It connects $v^{\alpha}_a$ (only) to $q^{\alpha}$ such that
$\dist(v^{\alpha}_a, q^{\alpha}) = M^2 + aM$.
See Figure~\ref{fig:np-hardness-distance} for an illustration.
%It connects $w^{\alpha}_a$ to $g_4$ using
%a simple path of length $M^2$.
%See Figure~\ref{fig:np-hardness}.
Consider the vertex in the simple path connecting
$q^{\alpha}$ and $u^{\alpha}_a$ which is adjacent with $q^{\alpha}$.
We denote the vertex by $q^{\alpha}_a$.
The reduction adds a pendant vertex and makes
it adjacent with $q^{\alpha}_a$.

The reduction repeats the constructions for elements in
$(\beta, b)$ and $(\gamma, c)$ in $\calU$.
In the first case, the vertices in the gadgets are connected to
$\{p^{\beta}, q^{\beta}, r^{\beta}\}$ whereas in the latter case
they are connected to $\{p^{\gamma}, q^{\gamma}, r^{\gamma}\}$.

Finally, for any vertices of the type $w^{\alpha}_a$, $w^{\beta}_b$, and
$w^{\gamma}_c$ added in the above steps, the reduction connects each of them to $g_4$ using simple paths of length $M^2$ for each connection.
See Figure~\ref{fig:np-hardness}.

\medskip

This completes the construction.
Suppose $P$ is the collection of all the pendant vertices
in $G$.
The reduction returns $(G, k = n + |P|)$ as
an instance of \textsc{Geodetic Set}.

\medskip

Note that each edge in Figure~\ref{fig:np-hardness}
corresponds to a simple path whose length is at least
$0.99 \cdot M^2$ and at most $1.01 \cdot M^2$.
Hence, the distance between any two vertices in
Figure~\ref{fig:np-hardness}, where we use an edge
to denote a simple path connecting two of its endpoints, is proportional to their distance in $G$.

Define $P$ to be the collection of all the pendants
vertices in $G$.
To prove the correctness of the reduction,
we first argue that vertices in $P$
covers most of the vertices in $G$.
We argue that the only vertices they do not cover
are internal vertices in the red edges of Figure~\ref{fig:np-hardness}.
We define these uncovered vertices as follows.
Define $V^{g_1}$ as the collection of all the internal
vertices in the simple path connecting $g_1$ to $z^i_1$ for any $i \in [n]$.
For every $a\in [n]$, define $V^{\alpha}_a$ as the collection
of internal vertices in the simple path connecting $q^{\alpha}$ to $u^{\alpha}_a$
and simple path connecting $u^{\alpha}_a$ to $w^{\alpha}_a$.
Define $V^{\alpha} = \bigcup_{a \in [n]} V^{\alpha}_a$.
Similarly, define $V^{\beta}$ and $V^{\gamma}$.
%Note that $V^{g_1} \cup V^{\alpha} \cup V^{\beta} \cup V^{\gamma}$
%is the collection of all the uncovered vertices.

Recall that $I(u, v)$ denotes the set of vertices in $G$ that are part of some shortest path between $u$ and $v$.
Also, we generalize this definition to set $S$ by
taking the union of $I(u, v)$ for all pairs of vertices $u, v$ in $S$.

\begin{lemma}
\label{lemma:pendant-cover-almost-everything}
Let $P$ be the collection of pendant vertices in $G$.
Then, $V(G) \setminus (V^{g_1} \cup V^{\alpha} \cup V^{\beta} \cup V^{\gamma}) = I(P)$.
\end{lemma}
\begin{proof}
We first prove that $V(G) \setminus (V^{g_1} \cup V^{\alpha} \cup V^{\beta} \cup V^{\gamma}) \subseteq I(P)$.
By the construction, any vertex $u$ in $G$ is adjacent
with at most one pendant vertex.
We use $\pndt(u)$ to denote the pendant vertex adjacent
with $u$, if it exists.
We first consider the partition of
vertices in $G$ based on simple paths they are part of.

\emph{Consider the simple paths ending at $g_1$}.
From the construction, $I(\pndt(g_1), \pndt(g_2))$
contains all the vertices in the simple path connecting
$g_1$ and $g_2$.
Similarly, $I(\pndt(z^i_2), \pndt(g_1))$
covers all the vertices in the simple path connecting
$z^i_2$ and $g_1$ for every $i \in [n]$.
Note that the statement of the lemma excludes the vertices
in simple paths connecting $g_1$ and $z^i_1$.

\emph{Consider the simple paths ending at $g_2$}.
From the construction, $I(\pndt(z^i_2), \pndt(g_2))$
covers all the vertices in the simple path connecting
$z^i_1$ to $g_2$ for every $i \in [n]$.
Consider vertices $x$ in $X^i$ and $y$ in $Y^i$
which are the part of $i^{th}$ set-encoding gadget
$\calS^i$ for some $i$ in $[n]$.
$I(\pndt(y), \pndt(g_2))$
covers all the vertices in the simple path connecting
$y$ to $g_2$.
Also, $I(\pndt(g_2), \pndt(g_3))$
covers all the vertices in the path connecting $g_2$ with $x_{p^{\alpha}}$.
A similar statement is true for the other branching vertices adjacent to $x$.

\emph{Consider the simple paths ending at $g_3$}.
From the construction, $I(\pndt(g_3), \pndt(g_4))$
contains all the vertices in the simple path connecting
$g_3$ and $g_4$.
As argued in the previous paragraph,
$I(\pndt(g_2), \pndt(g_3))$
covers all the vertices in the path connecting $g_3$ with
all the branching vertices like $x_{p^\alpha}$ that are
adjacent with $x$.

\emph{Consider the simple paths ending at $g_4$}.
For all such paths, its other endpoints are adjacent to some pendant 
vertex.
By the construction, the shortest path between the pendant
vertices adjacent to the endpoints covers all the vertices
in the simple path.

\emph{Consider the simple paths ending at $p^{\alpha}$,
$p^{\beta}$ or $p^{\gamma}$}.
All the simple paths ending at $p^{\alpha}$ can be partitioned
into two parts: ones that have the other endpoint in
a set-encoding gadget (i.e. $x_{p^{\alpha}}$ for some $x$
in $X^i$ for some $i \in [n]$)
or ones that have the other endpoint in
an element-encoding gadget (i.e. $u^{\alpha}_a$ for some
$a \in [n]$).
In the second case, both the endpoints
of the simple path are adjacent to pendant vertices.
Hence, by the construction, the shortest path between the pendant
vertices adjacent to the endpoints cover all the vertices
in the simple path.
In the first case, suppose the other endpoint of the simple path
is $x_{p^{\alpha}}$ for some $x$ in $X^i$.
Consider vertex $x$ and the unique vertex $y$ in $Y^i$ which is
at the distant $M^2$ from $x$.
By the construction, there is a shortest path from
$p^{\alpha}$ to $y$ that contains $x$.
Hence, $I(\pndt(p^{\alpha}), \pndt(y))$
covers all the vertices in the simple path
connecting $p^{\alpha}$ to $x$.
Since $x$ is an arbitrary point in $X^i$,
the statement is true for all the vertices in the first type of simple path.
These vertices are also covered by $I(\pndt(p^{\alpha}), \pndt(g_3))$.
A similar set of arguments implies that all the vertices in the simple path 
ending at $p^{\beta}$ and $p^{\gamma}$ are covered by vertices in $P$.

\emph{Consider the simple paths ending at $q^{\alpha}$,
$q^{\beta}$ or $q^{\gamma}$}.
Once again, all the simple paths ending at $q^{\alpha}$ can be partitioned
into two parts: ones that have the other endpoint in
a set-encoding gadget
(i.e. $x_{q^{\alpha}}$ for some $x$
in $X^i$ for some $i \in [n]$)
or ones that have the other endpoint in
an element-encoding gadget (i.e. $u^{\alpha}_a$ or $v^{\alpha}_a$ for some
$a \in [n]$).
Fix an integer $a \in [n]$ and recall that
$q^{\alpha}_a$ is adajcent with $q^{\alpha}$.
Using the same arguments as in the above paragraph,
$I(\pndt(q^{\alpha}_a), \pndt(y))$ covers all the vertices
in the simple path connecting
$q^{\alpha}$ and $x$.
This proves the claim for the vertices in the first type of path.
$I(\pndt(q^{\alpha}_ a), \pndt(u^{\alpha}_a))$ covers all the vertices
in the path connecting $q^{\alpha}_ a$ and $u_ a$.
Note that the vertices in the simple path connecting
$q^{\alpha}_ a$ and $v_ a$ are in $V^{\alpha}$,
which are excluded in the statement of the lemma.
A similar set of arguments implies that all the vertices in the simple path ending
at $q^{\beta}$ and $q^{\gamma}$ are covered by vertices in $P$.

\emph{Consider the simple paths ending at $r^{\alpha}$,
$r^{\beta}$ or $r^{\gamma}$}.
A similar set of arguments regarding paths ending at
$p^{\alpha}$, $p^{\beta}$ or $p^{\gamma}$ proves that
all the vertices in simple paths ending at $r^{\alpha}$,
$r^{\beta}$ or $r^{\gamma}$ are covered by vertices in $P$.

\emph{Consider the simple paths whose endpoints
are contained
in a set-encoding gadget}.
Consider a set-encoding gadget $\calS^i$ for a fixed $i$ in $[n]$.
The shortest paths connecting $z^i_2$ and $g_3$
covers all vertices in simple paths connecting $z^i_2$ to 
$z^i_1$ and $z^i_1$ to $x$ for every $x$ in $X^i$.
As argued before, there is a shortest path from
$p^{\alpha}$ to $y$, for any $y$ in $Y^i$, that contains the 
corresponding $x$ in $X^i$.
Hence, $I(\pndt(p^{\alpha}), \pndt(y))$
covers all the vertices in the simple path
connecting $x$ to $y$.
This implies that all the vertices in $\calS^i$,
except the vertices the vertices in $V^{g_1}$ are covered.
Recall that $V^{g_1}$ contains the internal vertices of the path
connecting $g_1$ with $z^i_1$.

\emph{Consider the simple paths whose both endpoints
are contained
in an element-encoding gadget}.
Consider an element-encoding gadget encoding
$(\alpha, a)$ for some $a$ in $[n]$.
By the construction, $I(\pndt(u^{\alpha}_a), \pndt(w^{\alpha}_a))$
covers all the vertices in the simple
path connecting $u^{\alpha}_a$ and $w^{\alpha}_a$.
We remark that the vertices in the simple path connecting
$v^{\alpha}_a$ and $w^{\alpha}_a$ are part of $V^{\alpha}$ and hence
excluded from the statement of the lemma.
Using similar arguments for elements of the form
$(\beta, b)$ and $(\gamma, c)$ for some $b, c$ in $[n]$.

The arguments above imply that
$V(G) \setminus (V^{g_1} \cup V^{\alpha} \cup V^{\beta} \cup V^{\gamma}) \subseteq I(P)$.
It remains to prove that no vertex in
$(V^{g_1} \cup V^{\alpha} \cup V^{\beta} \cup V^{\gamma})$
is covered by the vertices in $P$.

Consider the vertices in $V^{g_1}$ which
are in the $i^{th}$ copy of set-encoding gadget for
some $i$ in $[n]$.
It is easy to see that vertices in $P \setminus \{\pndt(z^i_2), \pndt(g_1)\}$
can not cover the vertices mentioned in the previous sentence.
Any shortest path whose one endpoint is $\pndt(z^i_2)$
and other endpoint is in $P \setminus \{\pndt(g_1)\}$ contains either
$\{g_1, g_2\}$ or $\{z^i_1, g_2\}$
Hence, none of such paths can cover vertices in the path connecting
$g_1$ and $z^i_1$.
Similarly, any shortest path whose one endpoint is $\pndt(g_1)$
and other endpoint is in $P \setminus \{\pndt(z^i_2)\}$ contains
$g_2$ and hence can not contain $z^i_1$.
This implies that $P$ can not cover any vertex in $V^{g_1}$

Consider the vertices in $V^{\alpha}$.
It is easy to verify that vertices in $V^{\alpha}$ are not covered
by the shortest path between any two common branching vertices.
Also, the shortest paths between
the pendant vertices in the element-encoding gadget do
not cover vertices in $V^{\alpha}$.
Consider vertices in $V^{\alpha}$, which are
part of the element-encoding gadget
corresponding to element $(\alpha, a)$ for some $a$ in $[n]$.
Any shortest path between the vertices in $P$
that can cover these vertices
has $\pndt(w^{\alpha}_a)$ as one of its endpoint and should contain $q^{\alpha}$.
By the construction, no path connecting $\pndt(w^{\alpha}_a)$
to any of the pendant vertex adjacent with common branching
vertices contains $q^{\alpha}$.
It is easy to see for $p^{\alpha}$, $r^{\alpha}$,
$g_3$ and $g_4$.
For the remaining common vertices, it is sufficient to prove
that the shortest path from  $\pndt(w^{\alpha}_a)$ to $g_2$
do not contain $q^{\alpha}$.
%Recall that $\dist(q^{\alpha}, v^{\alpha}_a) = M^2 + aM$ for
%some $a$ in $[n]$.
Note that the path from $w^{\alpha}_a$ to $g_2$ that contains
$q^{\alpha}$ and $x_{q^{\alpha}}$ for some $x$ in $X^i$,
is of length close to $4M^2$.
However, the shortest path from $\pndt(w^{\alpha}_a)$ to
$g_2$ is of length $3M^2$ and contains $g_3$ and $g_4$.
This implies that $P$ does not cover any vertex in $V^{\alpha}$.
Using similar arguments, the statement holds
for vertices in $V^{\beta}$ and $V^{\gamma}$
which concludes the proof of the lemma.
\end{proof}

In the following two lemmas, we prove that the reduction is safe.

\begin{lemma}
\label{lemma:forward-np-hardness-geodetic set}
If  $(\calU, \calS)$ is a \yes-instance of
\textsc{3-Dimensional Matching} then
$(G, k)$ is a \yes-instance of \textsc{Geodetic Set}.
\end{lemma}
\begin{proof}
Without loss of generality,
suppose $\{S_1, S_2, \dots, S_n\}$ is a collection
of $n$ sets in collection of $\calS$ (which contains $m$ sets)
that partitions all the vertices
in $\calU$.
For every $i$ in $[n]$,
consider a vertex in $X^i$ in the set-encoding
gadget $\calS^i$ corresponding to set $S_i$.
Define set $Q = \{x^1_1, x^2_2, \dots, x^n_n\}$,
where $x^i_i$ is in $X^i$.
Recall that $P$ is the collection of all the pendant vertices
in $G$.
It is easy to see that $|P \cup Q| = k$.
Due to Lemma~\ref{lemma:pendant-cover-almost-everything},
to prove that $P \cup Q$ is a geodetic set of $G$,
it is sufficient to prove that $Q$ covers vertices in
$ V^{g_1} \cup V^{\alpha} \cup V^{\beta} \cup V^{\gamma}$.
See the paragraph above Lemma~\ref{lemma:pendant-cover-almost-everything}
for the definition of these sets.
Also, recall that for any vertex $u$, $\pndt(u)$ denotes the unique pendant vertex adjacent to it, if one exists.

For a fixed $i$ in $[n]$,
$I(x^i_i, \pndt(g_1))$ covers all the vertices in
the simple path connecting $z^i_1$ to $g_1$.
As this is true for any $i$ in $[n]$,
$Q$ covers vertices in $V^{g_1}$.
Suppose $S^i = \{(\alpha, a), (\beta, b), (\gamma, c)\}$
for some $a, b, c$ in $[n]$.
Consider vertex $x^i_i$ in $X^i$.
We argue that the shortest paths between $x^i_i$ and $w^{\alpha}_a$
covers the vertices in $V^{\alpha}$ that are part of the element-encoding
gadget that is encoding element $(\alpha, a)$.
Consider Figure~\ref{fig:np-hardness-distance} with $a' = a$.
In this case, all three paths from $x^i_i$ to $u^{\alpha}_a$ via $p^{\alpha}$,
$q^{\alpha}$, and $q^{\alpha}$ as the same length of $2 M^2$.
This implies the shortest distance between $x^i_i$ and $w^{\alpha}_a$
is $3M^2$.
Moreover, the path connecting $x^i_i$ to $w^{\alpha}_a$
that contains $q^{\alpha}$ and $v^{\alpha}_a$ is shortest path between these
two vertices.
Hence, the shortest path between $x^i_i$ and $\pndt(w^{\alpha}_a)$
covers all the vertices in $V^{\alpha}$ that are part of element-encoding
gadget.
As $\{S_1, S_2, \dots, S_n\}$ partitions $\calU$,
(exactly) one of these $n$ set contains
element $(\alpha, a)$ for every $a$ in $[n]$,
all the vertices in $V^{\alpha}$
are covered by a vertex in $Q$ and a vertex in $P$.
Using the similar arguments, all the vertices in $V^{\beta}$ and
$V^{\gamma}$ are covered by vertices in $P \cup Q$.
This concludes the proof of the lemma.
\end{proof}

\begin{lemma}
\label{lemma:forward-np-hardness-geodetic set}
If $(G, k)$ is a \yes-instance of \textsc{Geodetic Set}
then $(\calU, \calS)$ is a \yes-instance of
\textsc{3-Dimensional Matching}.
\end{lemma}
\begin{proof}
Recall that $P$ is the collection of all the pendant vertices in $G$,
and $P$ is a subset of any geodetic set of $G$.
Suppose $P \cup Q$ is the geodetic set of $G$ of size of $k$,
and hence $|Q| \le n$.
By Lemma~\ref{lemma:pendant-cover-almost-everything},
vertices in $Q$ covers all the vertices in $V^{g_1} \cup V^{\alpha} \cup V^{\beta} \cup V^{\gamma} $.
See the paragraph before Lemma~\ref{lemma:pendant-cover-almost-everything}
for the definition of these sets.

We first prove that $Q$ contains at least one vertex in
the $i^{th}$ copy of set-encoding
gadget $\calS^i$ for every $i$ in $[n]$.
Assume, for the sake of contradiction, that this is not
the case for an index $i$.
Recall that $V^{g_1}$ is the union of internal vertices
connecting paths $g_1$ and $z^i_1$ for $i$ in $[n]$.
Define $V^{g_1}_i$ as the set of vertices in $V^{g_1}$ which
are in $\calS^i$.
Consider a vertex, say $h$, in $V^{g_1}_i$.
As $P \cup Q$ is a geodetic set of $G$, there are two vertices
in it,  say $s_1$ and $s_2$, that covers the vertex $h$.
By Lemma~\ref{lemma:pendant-cover-almost-everything},
$s_1$ and $s_2$ are \emph{not} in $P$.
By the above assumption, $s_1$ and $s_2$ are \emph{not} in
$\calS^i$.
Consider the shortest path from $s_1$ to $s_2$
that contains the vertex $h$.
Suppose that the paths enter and exist $\calS^i$
via common branching vertices $c_1$ and $c_2$.
Consider the case when $\{c_1, c_2\} \cap \{q^{\alpha}, q^{\beta},
q^{\gamma}\} = \emptyset$.
In this case, the shortest path between
$\pndt(c_1)$ and $\pndt(c_2)$ also contains the vertex $h$.
However, this contradicts Lemma~\ref{lemma:pendant-cover-almost-everything},
which states vertices in $P$ can not cover
any vertex in $V^{g_1}$.
Consider the other case when the shortest path
enters $\calS^i$ via, say, $q^{\alpha}$.
In this case, the above arguments can be
extended to the shortest path from $\pndt(q^{\alpha}_a)$
for some $a$ in $[n]$.
This also leads to contradiction to
Lemma~\ref{lemma:pendant-cover-almost-everything}.
Hence, our assumption is wrong and $Q$ contains at least one vertex in $\calS^i$.

As the cardinality of $Q$ is at most $n$,
this implies that $Q$ contains exactly one vertex in
$\calS^i$ for every $i$ in $[n]$.
Let $q$ be the unique vertex in $Q$ which is in $\calS^i$.
We determine the position of $Q$ in $\calS^i$ using the fact that
the shortest paths from $q$ to some other vertices
in $(P \cup Q) \setminus \{q\}$ covers all the uncovered 
vertices in the set-encoding gadget $\calS^i$.
We now restrict the possible cases for this other vertex in $P \cup Q$.
Let $h$ be the other vertex in $(P \cup Q) \setminus \{q\}$, so
the shortest path from $q$ to $h$ covers some vertices in $V^{g_1}_i$.
By the arguments above, $h$ is not in $\calS^i$.
%Suppose $h$ is not in $P$.
As argued in the previous paragraph,
the vertices covered by the shortest path from $q$ to $h$
are also covered by the shortest path from $q$ to $\pndt(c)$
where $c$ is a common branching vertex or
the vertex with pendant neighbor which is adjacent to either $q^{\alpha}$, $q^{\beta}$, or $q^{\gamma}$.
Hence, while considering the cases of shortest path from $q$ to $h$,
it is sufficient to check the cases when $h$ is one of the vertex
mentioned in the previous sentence.
Note that as no vertex in $Q$ is present in any
element-encoding gadget, vertices in $Q$, which are in set-encoding
gadgets, are also used to cover the uncovered vertices in the element-encoding
gadget.
We use this property to narrow the location of $q$ in $\calS^i$.

We partition $\calS^i$ into the following parts
and prove that $q$ can be located in particular types of parts.
%This location of vertices in $Q$ in $\calS^i$ naturally leads to
%a solution of \textsc{3-Dimentional Matching} instance.
By the construction,
deleting all the common branching vertices results in a
collection of trees.
Consider $\calS^i$ and root it at $z^i_1$.
\begin{itemize}[nolistsep]
\item For a vertex $x$ in $X^i$, consider the vertex $x_{p^{\alpha}}$
adjacent with it and on the path connecting $p^{\alpha}$.
Define \emph{$x_{p^{\alpha}}$-part} as the subtree rooted
at $x_{p^{\alpha}}$.
Similarly, define \emph{$x_{p^{\beta}}$-part} and \emph{$x_{p^{\gamma}}$-part}.
In Figure~\ref{fig:partition-set-encoding-gadget},
$x_{p^{\alpha}}$-part is highlighted in blue.
\item For every $y$ in $Y^i$, define $y$-part as subtree rooted
at $y$,
which is highlighted using orange color in
Figure~\ref{fig:partition-set-encoding-gadget}.
\item For every $x$ in $X^i$ and the
corresponding $y$ in $Y^i$, i.e. the unique vertex
in $Y^i$ which is at distance $M^2$ from $x$,
define \emph{$(x, y)$-part} as the
collection of internal vertices of the path from
$x$ to $y$.
See the brown edge in Figure~\ref{fig:partition-set-encoding-gadget}.
\item For every $x$ in $X^i$, define \emph{$(z^i_1, x]$-part}
as the vertex $x$ and the internal vertices of the path
from $x$ to $z^i_1$.
See the green part in Figure~\ref{fig:partition-set-encoding-gadget}.
Note that the part contains $x$ but not $z^i_1$.
\item Let $h, h_1$ and $h_2$ be the middle points of the paths
from $z^i_1$ to $z^i_2$, from $z^i_1$ to $g_1$, and
from $z^i_2$ to $g_2$, respectively.
Define \emph{$z^i_1$-part} as the tree rooted
at $z^i_1$ with $h, h_1$ and $h_2$ as its leaves.
This is highlighted by the purple color in Figure~\ref{fig:partition-set-encoding-gadget}.
\item All the remaining vertices are considered as \emph{left-over part},
highlighted using cyan color in Figure~\ref{fig:partition-set-encoding-gadget}.
Note that this is the only disconnected part in the partition.
\end{itemize}

\begin{figure}[t]
\centering
\includegraphics[scale=0.65]{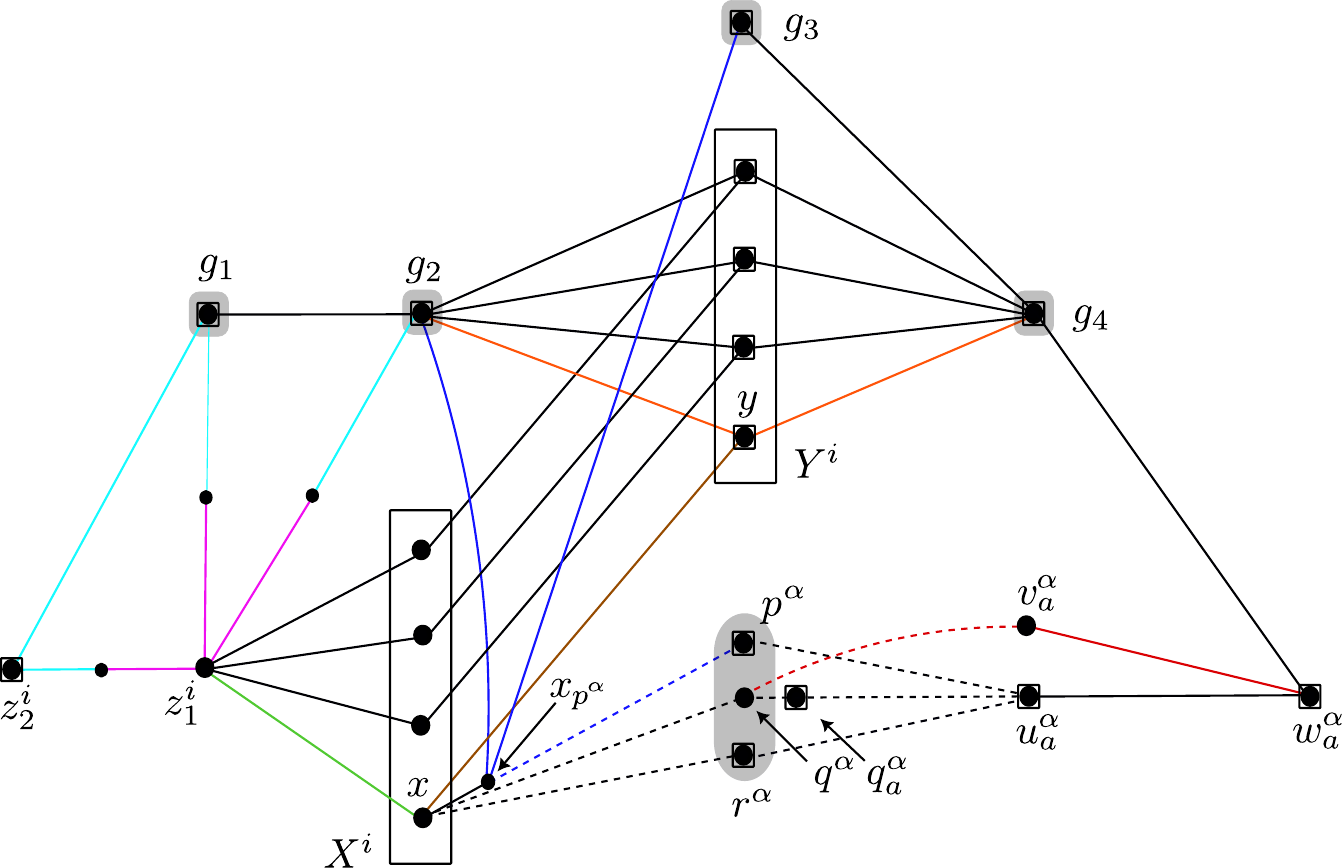}
\caption{Types of partition of vertices in $\calS^i$.
Different colors correspond to different types of parts in the set-encoding gadget.
For clarity, the figure does not show all the vertices in $\calS^i$.}
\label{fig:partition-set-encoding-gadget}
\end{figure}

Using the fact that $q$ is used to cover the vertices in $V^{g_1}_i$,
we first argue it is not in the part of the first-type or of
the second type mentioned above, nor can it be in the leftover part.
{Assume $q$ is in $x_{p^{\alpha}}$-part for some $x$ in $X^i$.}
In this case, the shortest path from $q$ to $g_1$ is via $g_2$
and hence does not contain vertices in $V^{g_1}_i$.
Similarly, the shortest path from $q$ to $z^i_2$ is via $g_1$ or via
$z^i_1$ and can not contain vertices in $V^{g_1}_i$.
It is easy to see that in this case, there are no paths from $q$ to any other vertex in $P$
can cover vertices in $V^{g_1}_i$.
Hence, $q$ cannot be in $x_{p^{\alpha}}$-part.
Now, {assume $q$ is in $y$-part for some $y$ in $Y^i$.}
The shortest path from $q$ to $g_1$ is via $g_2$
and to $z^i_2$ is via $g_1$ or $z^i_1$; hence, no such short paths can cover vertices in $V^{g_1}_i$.
%Once again, it is easy to see that in this case, no paths from $q$ to
%any other vertex in $P$ can cover vertices in $V^{g_1}_i$.
Finally, it is easy to see that there is no vertex in the leftover part
can cover any vertex in $V^{g_1}_i$.

We now turn our attention to the remaining parts, viz
$(x, y)$-part, $(z^i_1, x]$-part,
for some $x$ in $X^i$, or $z^i_1$-part.
Note that the shortest path from
any vertex in $(z^i_1, x]$-part or
$z^i_1$-part to $g_1$ covers all the vertices
in $V^{g_1}_i$ whereas the same is true
for half the vertices in $(x, y)$-part which
are closer to $x$ than they are to $y$.
In the next few paragraphs, we argue that the only
vertices in the $(z^i_1, x]$-part can be used
to cover the uncovered vertices in the element
encoding gadget.
Recall that for every $a\in [n]$, define $V^{\alpha}_a$ as
the collection of internal vertices in the simple path
connecting $q^{\alpha}$ to $u^{\alpha}_a$ and the simple path connecting $u^{\alpha}_a$ to $w^{\alpha}_a$.
We also defined $V^{\alpha} = \bigcup_{a \in [n]} V^{\alpha}_a$
and defined $V^{\beta}$ and $V^{\gamma}$ in a similar way.
Note that $V^{\alpha}$, $V^{\beta}$, and $V^{\gamma}$ are the collection of uncovered vertices in the element-encoding gadget.

Suppose vertex $x$ in $X^i$
corresponds to set $S = {(a', \alpha), (b', \beta), (c', \gamma)}$ and
$h$ is a vertex in $(z^i_1, x]$-part.
We prove that
$I(h, w^{\alpha}_a)$ covers $V^{\alpha}_a$ if and only if $a' = a$.
For a vertex $h$ in $(z^i_1, x]$-part, suppose
$\dist(h, x) = d_h$.
Consider the four paths from $h$ to $w^{\alpha}_a$
that are of length at most $d_h + M^2$
and contains exactly one common branching
point amongst $\{g_4, p^{\alpha}, q^{\alpha}, r^{\alpha}\}$.
Consider the path from $h$ to $w^{\alpha}_a$ that contains
$x$, the corresponding $y$, and $g_4$.
By the construction,
the length of this path is $d_h + 3M^2$.
For the remaining three paths,
consider Figure~\ref{fig:np-hardness-distance}.
The length of these three paths
is $d_h + 3M^2 + 2(a' - a)M$, $d_h + 3M^2 + (a - a')M$,
and $d_h + 3M^2 + 2(a - a')M$ when
it contains $p^{\alpha}$, $q^{\alpha}$, and $r^{\alpha}$,
respectively.
Note that $I(h, w^{\alpha}_a)$ covers $V^{\alpha}_a$ if and only
the shortest path from $h$ to $w^{\alpha}_a$ contains $q^{\alpha}$.
From the distances above, it is easy to see that this condition
is true only when $a = a'$.
Following the identical arguments, we can prove that
$I(h, w^{\beta}_b)$ covers $V^{\beta}_b$ if and only if $b' = b$, and
$I(h, w^{\gamma}_c)$ covers $V^{\gamma}_a$ if and only if $c' = c$.

We next prove that no vertex from $(x, y)$-part or $z^i_1$-part
can cover any uncovered vertex in the element-encoding gadget.
Consider a vertex $h$ in $(x, y)$-part.
As argued in the previous paragraph, $I(h, w^{\alpha}_a)$
can cover the vertices in $V^{\alpha}$ if and only if the shortest path
from $h$ to $w^{\alpha}_a$ contains $q^{\alpha}$.
However, the shortest path from $h$ to $w^{\alpha}_a$ contains $g_4$
and hence cannot cover the vertices in $V^{\alpha}$.
Now consider the vertices in $z^i_1$-part.
For any vertex $h$ in $z^i_1$-part which is an internal
vertex of path from $z^i_1$ to $g_2$, the shortest path
  $h$ to $w^{\alpha}_a$ contains $g_4$ and hence can not cover
vertices in $V^{\alpha}$.
Note that the shortest path from $z^i_1$ to $w^{\alpha}_a$ also
contains $g_4$ in it.
This is implied by the fact that the distance between $g_2$ and every
vertex in $Y^i $ is $M^2 - 1$ (instead of $M^2$ like most of the other paths).
Hence, $I(z^i_1, w^{\alpha}_a)$ does not cover any vertex
in $V^{\alpha}$.
For any vertex in the remaining $z^i_1$-part, the shortest path
from it to $w^{\alpha}_a$ is via $z^i$, which contains $g_4$.
These arguments imply that $q$ can not be in $(x, y)$-part or $z^i_1$-
part for any $x$ in $X^i$.

We are in a position to conclude the proof of the lemma.
Recall that $P \cup Q$ is a geodetic set of $G$.
Considering the position of vertices in $Q$, we construct a
collection of $n$ subsets $\calS'$ of $\calS$ as follows:
For every $i$ in $[n]$, let $x$ be the unique vertex in $X^i$
such that the vertex in $Q$ is in $(z^i_1, x]$-part of $\calS^i$.
If $x$ corresponds to set $S$,
then include $S$ in $\calS'$.
It is easy to see that the cardinality of $\calS'$ is at most $n$.
We argue that $\calS'$ covers all the vertices in $\calU$.
Consider an arbitrary element, say $(\alpha, a)$, in $\calU$.
The vertices in $V^{\alpha}$ are part of
the element-encoding gadget that encodes $(\alpha, a)$.
By above arguments, for some $i$ in $[n]$, there is $x$ in $X^i$ 
such that $(z^i_1, x_1]$-part of $\calS^i$
contains vertex $q$ in $Q$ with the property that
$I(q, w^{\alpha}_a)$ covers vertices in $V^{\alpha}$ in the element-encoding gadget.
Suppose $x$ corresponds to the vertex $S = \{(\alpha, a'), (\beta, b'), (\gamma, c')\}$ for some $a', b', c' \in [n]$.
From the arguments above, $I(q, w^{\alpha}_a)$ covers vertices in $V^{\alpha}$
if and only if $a' = a$.
This implies that $(\alpha, a)$ is covered by some set in $\calS'$.
Since this is an arbitrary element in $\calU$, one can conclude
that $\calS'$ covers all the vertices in $\calU$.
This implies that $(\calU, \calS)$ is a \yes-instance of \textsc{3-Dimensional Matching}.
\end{proof}

\begin{lemma}
Pathwidth and feedback vertex set the number of
$G$ are at most $17$ and $13$, respectively.
\end{lemma}
\begin{proof}
We present
a mixed search strategy to clean $G$ using $17$ searchers
to bound the pathwidth of $G$.
Place $13$ searchers on common branching vertices.
These searchers will never move from their places.
(This is equivalent to deleting common branching vertices and presenting
a mixed search strategy using four searchers for the resulting graph.)
We search the graph in $3n + n$ rounds,
one round for each element-encoding gadget and
one round for each set of set-encoding gadget, respectively.

Consider an element-encoding gadget encoding
an element, say $(\alpha, a)$ for some $a$ in $[n]$.
We can place three searchers on $u^{\alpha}_a$, $v^{\alpha}_a$,
and $w^{\alpha}_a$, and the fourth searcher can move along the simple
the path connecting these vertices and with common branching vertices
to clean all the edges in the connecting paths.

By the construction, deleting all the common branching vertices results
in a collection of trees.
Consider the set-encoding gadget $\calS^i$ for some $i$ in $[n]$ and
root it at $z^i_1$.
In each round, we fix a searcher on $z^i_1$.
We need one searcher to clean all the edges between the simple
path connecting $z^i_1$ to $g_1$ via $z^i_2$,
path of length $M^2$ connecting $z^i_1$ to $g_1$,
and finally path connecting $z^i_1$ to $g_2$.
Once these paths are cleaned, the searcher is free
and can be used to clean the remaining part of the
set-encoding gadget.
We divide the remaining part of this round
into $m$ sub-rounds corresponding
to each vertex in $X^i$.
For a sub-round, fix two searchers on a vertex $x$ in $X^i$
and another searcher on the corresponding vertex $y$ in $Y^i$
which is at distance $M^2$ from $x$.
The remaining searcher can clean all the
edges in the simple paths connecting $x$ to $y$, $y$ to common
vertices, and $x$ to $z^i_1$.
Next, the searcher placed on $y$ can be relocated to
each of the branching vertices adjacent to $x$, i.e. vertices of the form
$x_{p^{\alpha}}$.
Now, the other searcher can be used to clean the subtree rooted at
the branching vertex.
This completes one sub-round.
Completing such $m$ sub-rounds clears
all the paths in $\calS^i$.

As we need $17$ searchers to clean the graph,
result of Takahashi et al.~\cite{DBLP:journals/tcs/TakahashiUK95}
implies that $\pw(G) \le 17$.

As deleting all the common branching vertices results in a
collection of trees, feedback vertex set number of the graph is $13$.
\end{proof}

% !TEX root = ./main.tex
\section{Conclusion}\label{sec:conclu}

In this article, we proved that \textsc{Geodetic Set}
is \NP-complete even on graphs with constant pathwidth and feedback
vertex set number answering the open question by Kellerhals and Koana~\cite{KK22}.
In the same paper, the authors provided an \FPT\ algorithm
(with impractical running time) when parameterized by
treedepth (denoted by $\td$) of the input graph.
Recent work by Foucaud et al.~\cite{DBLP:conf/icalp/FoucaudGK0IST24}
showed that \textsc{Geodetic Set} admits an algorithm running in time
$2^{\diam^{\calO(\tw)}} \cdot poly(|V(G)|)$,
where $\diam$ denotes the diameter of the graph.
This result, along with the fact that $\diam$ and $\tw$ are upper bounded
by $2^{\td}$ and $\td + 1$, respectively,
imply that \textsc{Geodetic Set}
admits an algorithm running in time $2^{2^{\calO(\td^2)}}\cdot poly(|V(G)|)$.
It would be interesting to obtain an \FPT\ algorithm
parameterized by treedepth with improved running time.
We remark that
Foucaud et al.~\cite{DBLP:conf/icalp/FoucaudGK0IST24} proved
that the problem does not admit an algorithm running in time
$2^{2^{o(\td)}}\cdot poly(|V(G)|)$, unless the \ETH\ fails.
%Note that both \gsfull and \textsc{Metric Dimension} are \NP-hard
%for interval graphs~\cite{floISAAC20,FoucaudMNPV17b}, and those
%graphs have treelength~$1$.
%Whether \gsfull is \FPT\ when parameterized by the solution size
%on interval graphs is an open question raised in~\cite{floISAAC20}.

As highlighted in the introduction of this article and
also in ~\cite{floISAAC20,KK22},
\mdfull and \gsfull share many hardness properties.
We mention one notable exception:
\textsc{Metric Dimension} is \FPT\ when
parameterized by the solution size plus the treelength~\cite{BelmonteFGR17},
however, \gsfull\ is \W[2]-hard when parameterized by 
the solution size even when treelength is a constant.
See \cite[Theorem~2]{DBLP:journals/dm/DouradoPRS10}.
Regarding the parameterized complexity
of \mdfull\ when parameterized by pathwidth
and feedback vertex set number of the graph,
Galby et al.~\cite{GKIST23} showed that the problem is
\W[1]-hard.
Can we improve this result to obtain a result for \mdfull\
as we did for \gsfull\ in this article?

%\newpage

%\input{np-hardness-treewidth-smd}

\bibliography{bib}

%\appendix

\end{document}